\documentclass[journal,twoside,web]{ieeecolor}

\usepackage{amsmath} %
\usepackage{amssymb}  %

\usepackage{amsthm}
\usepackage{xcolor}
\usepackage{mathtools,stmaryrd}
\usepackage{empheq}
\usepackage[hidelinks]{hyperref}
\usepackage{algorithm}
\usepackage{algpseudocodex}
\let\labelindent\undefined
\usepackage{enumitem}
\usepackage{lipsum}
\usepackage{subcaption}
\usepackage{booktabs}
\usepackage{caption}
\usepackage{xspace}
\usepackage{lcsys}

\newcommand{\immrax}{\texttt{immrax}}

\newcommand{\ie}{\emph{i.e.}}

\newcommand{\smallconc}[2]{\begin{bsmallmatrix} #1 \\ #2 \end{bsmallmatrix}}

\newcommand{\co}{\operatorname{co}}

\newcommand{\diag}{\operatorname{diag}}

\newcommand{\supress}[1]{}

\newcommand{\R}{\mathbb{R}}

\newcommand{\IR}{\mathbb{IR}}

\newcommand{\ul}[1]{\underline{#1}}
\newcommand{\ula}{\ul{a}}

\newcommand{\ulA}{\ul{A}}

\newcommand{\ol}[1]{\overline{#1}}
\newcommand{\ola}{\ol{a}}

\newcommand{\olA}{\ol{A}}

\newcommand{\olR}{\ol{R}}

\newcommand{\rmd}{\mathrm{d}}

\theoremstyle{plain}

\newtheorem{lemma}{Lemma}
\newtheorem{theorem}{Theorem}
\newtheorem{corollary}{Corollary}

\theoremstyle{definition}
\newtheorem{definition}{Definition}
\newtheorem{example}{Example}

\theoremstyle{remark}
\newtheorem{remark}{Remark}

\title{
Learning Certified Neural Network Controllers Using Contraction and Interval Analysis
}

\author{%
Akash Harapanahalli$^1$, Samuel Coogan$^1$, and Alexander Davydov$^2$ %
\thanks{*This work was supported in part by the National Science Foundation under awards \#2219755 and \#2440387, and by the Air Force Office of Scientific Research under Grant FA9550-23-1-0303.}%
\thanks{$^{1}$Akash Harapanahalli and Samuel Coogan are with the School of Electrical and Computer Engineering, Georgia Institute of Technology, Atlanta, GA, 30332, USA. \{aharapan,sam.coogan\}@gatech.edu}%
\thanks{$^{2}$Alexander Davydov 
is with the Department of Mechanical Engineering and Ken Kennedy Institute, Rice University, Houston, TX, 77005, USA. davydov@rice.edu
}%
}

\newcommand{\fol}{f_{\mathrm{ol}}}
\newcommand{\fpi}{f_\pi}
\newcommand{\fcl}{f_{\mathrm{cl}}}
\newcommand{\fd}{f_{\mathrm{d}}}
\newcommand{\Dfx}{\frac{\partial \fpi}{\partial x}}

\newcommand{\intX}{\operatorname{int}X}
\newcommand{\Ac}{A^c}
\newcommand{\Ap}{A^\Delta}
\newcommand{\Gc}{G^c}
\newcommand{\Gp}{G^\Delta}

\newcommand{\change}[1]{#1}

\begin{document}

\maketitle
\thispagestyle{empty}
\pagestyle{empty}

\begin{abstract}
We present a novel framework that jointly trains a neural network controller and a neural Riemannian metric with rigorous closed-loop contraction guarantees using formal bound propagation.
Directly bounding the symmetric Riemannian contraction linear matrix inequality causes unnecessary overconservativeness due to poor dependency management.
Instead, we analyze an asymmetric matrix function $G$, where $2^n$ GPU-parallelized corner checks of its interval hull verify that an entire interval subset $X$ is a contraction region in a single shot.
This eliminates the sample complexity problems encountered with previous Lipschitz-based guarantees.
Additionally, for control-affine systems under a Killing field assumption, our method produces an explicit tracking controller capable of exponentially stabilizing any dynamically feasible trajectory using just two forward inferences of the learned policy.
Using JAX and \immrax{} for linear bound propagation, we apply this approach to a full 10-state quadrotor model.
In $<$10 minutes of post-JIT training, we simultaneously learn a control policy $\pi$, a neural contraction metric $\Theta$, and a verified 10-dimensional contraction region $X$.
\end{abstract}
\begin{IEEEkeywords}
Neural networks, contraction theory, formal methods
\end{IEEEkeywords}

\section{Introduction}

\IEEEPARstart{L}{earning}-based methods have demonstrated potential for controlling complex nonlinear systems, owing to their ability to represent and learn highly expressive policies that can be difficult to obtain through classical control design. However, deploying learned controllers in safety-critical applications demands rigorous guarantees on closed-loop behavior that such policies do not inherently provide. 

This gap between expressiveness and certifiability has motivated a growing body of work on formally verifying closed-loop behavior of systems with neural network controllers.
For instance, one can certify neural Lyapunov functions for stabilization or neural control barrier functions for safety, using satisfiability modulo \change{theories (SMT)} or branch-and-bound optimization to rigorously conclude the lack of counterexamples \cite{abate_formal_2021}.
These exact verification approaches, however, can be difficult to scale to large state spaces and deeper networks.
A scalable alternative uses \emph{bound propagation} \cite{liu_algorithms_2021} to verify guaranteed overapproximations of the learned components.

Beyond stabilization and invariance, contraction theory has provided a suite of tools suitable to tackle other challenges like online tracking \cite{manchester_control_2017}, and recent work has developed neural contraction metrics \cite{sun_learning_2021} to help scale these methods.
However, formally verifying these contraction metrics remains a significant challenge, relying on Lipschitz-based sampling arguments or over-conservative bounding techniques \cite{zakwan_neural_2024,davydov_verifying_2025}.
In this work, we use off-the-shelf neural network bound propagators, careful dependency management through the contraction condition, and an efficient exact characterization of the $\ell_2$-logarithmic norm of an interval matrix to jointly train neural network controllers and contraction metrics with formal guarantees.

\subsection{Problem Formulation and Contributions}

Consider a nonlinear control system of the form
\begin{align} \label{eq:fol}
    \dot{x} = \fol(t,x,u),
\end{align}
where $x\in\R^n$ is the state of the system, $u\in\R^m$ is a control input, and $\fol:[0,\infty)\times\R^n\times\R^m\to\R^n$ is a $C^1$ smooth mapping.
Our primary goal is to jointly synthesize an explicit feedback policy $\pi(x)$ and a Riemannian metric $M(x)$ such that the closed-loop system
\begin{align} \label{eq:fpi}
    \dot{x} = \fpi(t,x) := \fol(t,x,\pi(x))
\end{align}
is contracting. 
We also consider the special case of control-affine systems, $\dot{x} = \fd(x) + g(x)u$, where the goal extends to tracking any dynamically feasible trajectory.

To address these problems, we develop a novel framework using bound propagation to jointly train a neural contraction metric and a neural network controller for nonlinear systems with formal guarantees of closed-loop contractivity.
Our specific contributions are outlined:
\newcommand{\contrtitle}[1]{\textbf{#1}.}
\begin{itemize}[left=0pt]
    \item \contrtitle{Asymmetric Contraction Condition} We propose an \change{asymmetric Riemannian contraction condition, equivalent to the standard linear matrix inequality (LMI) provided in \cite{lohmiller_contraction_1998}, which bounds} the $\ell_2$-logarithmic norm of a matrix directly involving the factored parameterization of the metric. 
    While mathematically equivalent on point-wise evaluations, our asymmetric condition can dramatically reduce overestimation error when using bound propagation tools due to better dependency management.
    \item \contrtitle{Certification and Training via Bound Propagation} Using off-the-shelf bound propagators with automatic differentiation, we algorithmically obtain an interval hull of this asymmetric matrix. 
    Applying a result from \cite{rohn_positive_1994}, we show how the maximum logarithmic norm within this interval hull can be checked with $2^n$ GPU parallelized corner checks. We design a loss function in Algorithm \ref{alg:loss} which, if nonpositive, certifies the contraction condition over an entire region $X$ without exhaustive sampling.
\item \contrtitle{Explicit Tracking Control} In the special case where the system is control affine and the metric satisfies a Killing field condition (matching the strong \change{control contraction metric} (CCM) condition from \cite[C2]{manchester_control_2017}), we propose a simple, explicit feedback controller $u(t,x) = \pi(x) - \pi(x'(t)) + u'(t)$, capable of tracking any \change{open-loop} dynamically feasible \change{pair $(x'(t),u'(t))$} using only two forward inferences of the learned policy $\pi$.
    Notably, our method eliminates the need to search for a geodesic of the metric during online execution, and does not require any state augmentation into error coordinates.
\end{itemize}

\change{One of the crucial distinctions between our approach and previous verification frameworks such as~\cite{fitzsimmons_computation_2024}, which certifies neural contraction metrics after training using SMT solvers, is that we embed a differentiable certificate from the efficient bound propagator CROWN \cite{zhang_efficient_2018} directly in the training loss.
This allows us to \emph{learn} controllers and neural contraction metrics that are certified by construction, rather than \emph{post hoc} verification, and then optimize directly against the overestimation error to help mitigate wrapping effects. }

Finally, we \change{illustrate} through a 10-state quadrotor benchmark how the proposed method can certify and train explicit contracting tracking controllers, where sampling-based certification would suffer from the curse of dimensionality.

\subsection{Background, Notation, and Related Work}

\paragraph*{Riemannian Metrics and Contraction Theory}

Contraction theory provides a geometric framework for translating infinitesimal linear properties into incremental exponential stability bounds between trajectories of the system.
There is a rich literature in contraction analysis, including \cite{lohmiller_contraction_1998} for Riemannian spaces, \cite{forni_differential_2013} for a differential Lyapunov function oriented framework, \cite{davydov_non-euclidean_2022} for distances defined by possibly non-Euclidean norms, and \cite{sontag_contractive_2010,manchester_control_2017} for controller design.
In this work, we focus on Riemannian geometry in $\R^n$.

Let $\mu_2(\cdot)$ denote the induced $\ell_2$ logarithmic norm $\mu_2(A) = \lambda_{\max}(\tfrac{A + A^T}{2})$.
For a matrix function $A:\R^n\to\R^{n\times n}$, let $\partial_{v}A(x) = \sum_{i=1}^n \frac{\partial A}{\partial x_i}(x) v_i$ denote the directional derivative.
Let $X\subseteq\R^n$ be a set whose interior $\intX$ is nonempty and connected. 
A \emph{Riemannian metric} $M$ on $X$ is a smooth, symmetric, positive-definite matrix function that is \emph{uniformly bounded}, meaning there exist $a,b>0$ such that $aI \preceq M(x) \preceq bI$ for every $x\in X$.

$M$ defines a length for piecewise-smooth paths $\gamma:[0,1]\to\intX$ by integrating the Riemannian norm of the velocity: $\ell(\gamma) = \int_0^1 \sqrt{\gamma'(s)^T M(\gamma(s)) \gamma'(s)} \rmd s$.
This equips $\intX$ with a distance $d(x,y)$, defined as the infimum of the lengths of all such paths in $\intX$ connecting $x$ and $y$ (we do not assume $\intX$ is geodesically complete, so this infimum may not be achieved by a geodesic within the space).

\begin{definition}[Contraction Region {\cite[Def. 2]{lohmiller_contraction_1998}}] \label{def:contr_region}
A region of the state space $X$ is called a \emph{contraction region} of the vector field $\fpi$ from \eqref{eq:fpi} with respect to \change{uniformly bounded} metric $M(x)$ if there exists $c>0$ such that the linear matrix inequality $\change{S(t,x;c)} \preceq 0$ holds uniformly over $t\in[0,\infty)$, $x\in X$, where 
\begin{align} \label{eq:contraction_LMI}
\begin{aligned}
    \change{S(t,x;c)} := M(x) & \Dfx(t,x) + \Dfx(t,x)^T M(x) \\
    &+ \partial_{\fpi(t,x)} M(x) + 2cM(x).
\end{aligned}
\end{align}
\end{definition}
We note that Definition \ref{def:contr_region} does not require $X$ to be forward invariant.
Definition \ref{def:contr_region} provides an important infinitesimal characterization of the following property: 
if two trajectories \change{$x,x'$ of $\fpi$ \eqref{eq:fpi}, with initial conditions $x_0,x'_0$ respectively at time $t_0$,} are contained within $\intX$ \change{and} have a minimizing geodesic connecting them fully contained within $\intX$, then
\begin{align*}
    d(x(t),x'(t)) \leq e^{-c\change{(t - t_0)}} d(x_0,x'_0).
\end{align*}
This bound is useful in a number of contexts---for instance, for autonomous systems, if a contraction region is forward invariant, geodesically convex and complete, the flow map $\phi_\tau$ for fixed $\tau>0$ is a contraction mapping, and Banach's fixed point theorem concludes that the system exponentially converges to a unique equilibrium.

\begin{table}[]
    \centering
    \begin{tabular}{cl}
        \toprule
        $\fol$ & Open-loop system \eqref{eq:fol} \\
        $\fpi$ & Closed-loop system \eqref{eq:fpi} under policy $u = \pi(x)$ \\
        $\fd$  & Drift term for the control-affine system from \eqref{eq:control_affine} \\
        $g$    & Input vector fields for the control-affine system from \eqref{eq:control_affine} \\
        $\fcl$ & Closed-loop system \eqref{eq:fcl} under policy $u = \pi(x) - \pi(x') + u'$ \\
        $X$ & Region of the state space $\R^n$ with nonempty interior $\intX$ \\
        $M$ & Riemannian metric over $X$ \\
        $\Theta$ & Upper triangular factorization satisfying $M(x) = \Theta(x)^T\Theta(x)$ \\
        $d$ & Metric geodesic distance associated to $M$ \\
        $\pi$ & Feedback policy \\
        $S$ & Left-hand side of the contraction LMI \eqref{eq:contraction_LMI} \\
        $G$ & Asymmetric contraction matrix from Definition \ref{def:asym_contr_mtx} \\
        \bottomrule
    \end{tabular}
    \caption{Symbols and Notation}
    \label{tab:placeholder}
    \vspace{-2.1em}
\end{table}

\paragraph*{Interval Analysis and Bound Propagation} \label{sec:bound_prop}

Let $\IR$ denote the set of closed intervals of $\R$, of the form $[\ula,\ola] := \{a\in\R : \ula \leq a \leq \ola\}$ for $\ula,\ola\in\R$. 
An interval matrix $[A] = [\ulA,\olA] \in\IR^{m\times n}$ is a matrix where $[A]_{ij} = [\ulA_{ij},\olA_{ij}]\in\IR$. $[A]$ is a subset of $\R^{m\times n}$ in the entrywise sense, so $A\in[A]$ if $A_{ij}\in[A]_{ij}$ for every $i=1,\dots,m$ and $j=1,\dots,n$.
For the interval matrix $[A]$, we notate its center $\Ac\in\R^{m\times n}$ as $\Ac = \big(\ulA + \olA\big)/{2}$ and its radius $\Ap\in\R^{m\times n}$ as $\Ap = \big(\olA - \ulA\big)/{2}$.
\change{For a vector $v\in\R^n$, let $\diag(v)\in\R^{n\times n}$ be the matrix $\diag(v)_{ii} = v_i$ and $\diag(v)_{ij} = 0$ for $i\neq j$.}

Given a matrix function $A:\R^n\to\R^{m\times n}$, an \emph{interval hull} of $A$ over an input set $X\subseteq\R^n$ is an interval matrix $[A]\in\R^{m\times n}$ such that $A(x) \in [A]$ uniformly for every $x\in X$.
Obtaining the smallest interval hull of $A$ is generally an intractable nonconvex optimization problem.
Instead, a general way to obtain larger valid interval hulls is through \emph{bound propagation}, where (i) a computation tree for $A$ is built through an abstract evaluation, and (ii) the input bound $X$ is successively propagated through the computation tree.
These existing approaches have varying tradeoffs of accuracy and runtime: 
interval bound propagation \cite{jaulin_applied_2001,gowal_effectiveness_2018}, which efficiently propagates lower and upper scalar bounds; 
linear bound propagation \cite{zhang_efficient_2018}, which propagates linear lower and upper bounding functions;
set-based methods \cite{althoff_set_2021} including zonotopes and star sets \cite{tran_star-based_2019}, which propagate set overapproximations forward through the network;
Taylor models \cite{chen_reachability_2015} and other polynomial approximators, which propagate tighter but more computationally expensive representations.

A fundamental issue often encountered by these approaches is the \emph{dependency problem}.
When a variable appears multiple times in a computation graph, these methods treat each instance as an independent variable, failing to capture algebraic interactions between terms. 
We demonstrate this phenomenon in Example \ref{ex:dependency_management} in the next Section, and show how our asymmetric formulation mitigates these issues.
\section{Riemannian Contraction via Bound Propagation}

In prior contraction synthesis works \cite{sun_learning_2021}, the LMI \eqref{eq:contraction_LMI} is incorporated into the training by sampling points $x_n$ across the desired contraction region $X$ and backpropagating against $\lambda_{\max}(\change{S(t,x_n;c)})$.
In this section, we use bound propagation to directly verify the contraction condition over the region $X$ without sampling, which we use to design an explicit tracking control policy for control-affine systems.

\subsection{The Asymmetric Contraction Matrix} \label{sec:asym_contr_mtx}

We first propose an asymmetric matrix that reduces the overestimation error from bound propagation.

\begin{definition} [Asymmetric Contraction Matrix] \label{def:asym_contr_mtx}
Consider the vector field $\fpi$ from \eqref{eq:fpi}, a Riemannian metric \change{$M(x) = \Theta(x)^T \Theta(x) + aI$} for $a\geq 0$, and a fixed $c>0$.
Define the \emph{asymmetric contraction matrix} as the following:
\begin{align} \label{eq:asym_contr_mtx}
\begin{aligned}
    &\change{G(t,x;a,c) := a\left(\Dfx(t,x) + cI\right)} \\
    & +\Theta(x)^T \left[\partial_{\fpi(t,x)} \Theta(x) + \Theta(x) \left(\Dfx(t,x) + cI\right)\right] \\
\end{aligned}
\end{align}
\end{definition}

The following lemma demonstrates how an $\ell_2$-logarithmic norm bound on $G$ is equivalent to the LMI \eqref{eq:contraction_LMI}, building off the fact that the symmetrization of $G$ exactly equals $\tfrac12 S$.

\begin{lemma}[Equivalence] \label{lem:uniform_equivalence}
Consider the vector field $\fpi$ from \eqref{eq:fpi}, let \change{$M(x) = \Theta(x)^T\Theta(x) +aI$ for $a\geq 0$} be a \change{uniformly bounded} Riemannian metric, and $c>0$ be fixed.
The following are equivalent:
\begin{enumerate}[label=(\roman*)]
    \item $\change{S(t,x;c)} \preceq 0$, with $S$ defined in \eqref{eq:contraction_LMI};
    \item $\mu_2(\change{G(t,x;a,c)}) \leq 0$, with $G$ defined as the asymmetric contraction matrix \eqref{eq:asym_contr_mtx}.
\end{enumerate}
Thus, $X$ is a contraction region if and only if $\mu_2(\change{G(t,x;a,c)}) \leq 0$ uniformly over $t\in[0,\infty)$, $x\in X$.
\end{lemma}
\begin{proof}
Expanding $G$ and adding its transpose reveals
\begin{align*}
    &\change{G(t,x;a,c)} + \change{G(t,x;a,c)}^T  = 2c\Theta(x)^T \Theta(x)  \\
    &+ \textstyle \Theta(x)^T \partial_{\fpi(t,x)} \Theta(x) + \left(\partial_{\fpi(t,x)} \Theta(x)\right)^T \Theta(x) \\
    &+ \textstyle \Theta(x)^T \Theta(x) \Dfx(t,x) + \Dfx(t,x)^T \Theta(x)^T\Theta(x) \\
    & \textstyle \change{+ a\Dfx(t,x) + a\Dfx(t,x)^T +  2acI}
\end{align*}
Substituting $M(x) = \Theta(x)^T \Theta(x) \change{+ aI}$ and $\partial_{\fpi(x)}M(x) = \left(\partial_{\fpi(x)}\Theta(x)\right)^T \Theta(x) + \Theta(x)^T \partial_{\fpi(x)}\Theta(x)$ using the matrix product rule shows that $\change{S(t,x;c)} = \change{G(t,x;a,c)} + \change{G(t,x;a,c)}^T$. 
Thus, $\mu_2(\change{G(t,x;a,c)}) = \lambda_{\max}\left(\frac{\change{S(t,x;c)}}{2}\right)\leq 0$ if and only if $\change{S(t,x;c)}$ has nonpositive eigenvalues.
\end{proof}

We note that the condition $\mu_2(\change{G(t,x;a,c)}) \leq 0$ is different from the condition $\mu_2(F(t,x)) \leq -c$ provided in \cite[Def. 2]{lohmiller_contraction_1998}, which bounds the logarithmic norm of the \emph{generalized Jacobian} $F(t,x) = (\partial_{\fpi(t,x)} \Theta(x) + \Theta(x) \Dfx(t,x))\Theta(x)^{-1}$ \change{(for the $a=0$ case)}. 
Bounding $F$ would require propagating bounds through a matrix inversion, which is difficult.
Instead, the asymmetric contraction matrix \eqref{eq:asym_contr_mtx} requires only a transpose of $\Theta$.

The following analytical example demonstrates how while the two conditions are theoretically equivalent, applying bound propagation strategies to $\change{G(t,x;a,c)}$ instead of $\change{S(t,x;c)}$ can dramatically reduce overestimation error due to lost information in the symmetrization.

\begin{example}[Dependency management] \label{ex:dependency_management}
In the scalar case, we can easily compare the expressions \eqref{eq:contraction_LMI} and \eqref{eq:asym_contr_mtx},
\begin{align*}
    S &= \Theta^2 J + J\Theta^2 + \Theta\dot{\Theta} + \dot{\Theta}\Theta + 2c\Theta^2, \\
    G &= \Theta(\dot{\Theta} + \Theta (J + c)),
\end{align*}
\change{with the notation $\dot{\Theta} = \partial_{\fcl}\Theta$, and $a=0$.}
\change{Since $S = 2G$, this may seem like a simple notational difference}. 
However, recall from Section \ref{sec:bound_prop} that bound propagators like interval analysis and CROWN directly trace the above expressions into computation graphs---\change{the resulting computations and uncertainty propagation can yield significant differences between $S$ and $G$.}
The expression for $G$ captures interactions between $\dot{\Theta}$ and $\Theta J$ before premultiplying by $\Theta$, while the expression for $S$ multiplies \change{first, introducing redundancies}.

For instance, evaluation of the above expressions over $\Theta\in[0.5,1]$, $\dot{\Theta}\in[-2,-1.5]$, $J\in[-1,1]$ and $c = 0.5$ gives
\begin{align*}
    S \in [-5.75,1.5] \quad \text{ and } \quad 2G\in[-5,0],
\end{align*}
which is a significant difference when trying to verify nonpositivity (contraction) using Lemma \ref{lem:uniform_equivalence}.
\end{example}

\subsection{Efficient Parallel Corner Checks}

Given an interval matrix $[A]$, one way to solve $\max_{A\in[A]}\mu_2(A)$ is to (i) write $[A] = \co\{A_k\}_{k=1}^{2^{n^2}}$ ($\co$ is the convex hull), where the $A_k$ consist of all $2^{n^2}$ matrices of the form $(A_k)_{ij} \in \{\ulA_{ij}, \olA_{ij}\}$, and (ii) note that $\mu_2$ is convex, so $\max_{A\in[A]}\mu_2(A) = \max_{k \in \{1,\dots,2^{n^2}\}}\mu_2(A_k)$. 
While this exactly computes the largest logarithmic norm of the set, it scales very poorly in $n$.
Instead, we apply a result from \cite{rohn_positive_1994} to reduce the number of corner checks down to $2^n$.
For the $n=10$ quadrotor in Section \ref{sec:example}, for instance, $2^{10^2}$ is over $10^{30}$ corners, while $2^{10}$ is far more practical at $1024$ corners.

\begin{lemma}[{\cite[Thm 1]{rohn_positive_1994}}] \label{lem:corner_check}
For an interval matrix $[A]\in\IR^{n\times n}$,
\begin{align*}
    &\max_{A\in[A]} \mu_2(A) \\
    &= \max \left\{ \mu_2 (\Ac + \diag(s)\Ap \diag(s)) : s\in\{-1,+1\}^n \right\}.
\end{align*}
\end{lemma}
The proof is in \cite{rohn_positive_1994} in the context of verifying the positive definiteness of an interval matrix---we make slight modifications to bound the function $\mu_2(A) = \lambda_{\max}(\frac{A + A^T}{2})$ instead.
\change{The essential idea is to decompose $v^T Av$ into $v^T \Ac v$ and bound the difference using $|v|^T \Ap |v|$,  where $|\cdot|$ is the entrywise absolute value.
Then, we consider possible choices of $|v|$ instead of possible choices of $A - \Ac$.}
\begin{proof}
We first recall $\mu_2(A) = \max_{v : |v|_2 = 1} (v^T \tfrac{A + A^T}{2} v) = \max_{v : |v|_2 = 1} v^T A v$.
Let $A\in[A]$. We can write $v^T A v = v^T \Ac v + v^T (A - \Ac) v$. 
\change{Note that for any row $a^T$ and vector $v$, $|a^T v| \leq |a^T||v|$.}
Thus, $|v^T(A - \Ac) v|\leq \change{|v^T(A - \Ac)||v| } \leq |v|^T |A - \Ac| |v| \leq |v|^T A_{\Delta} |v|$, and therefore $v^T Av \leq v^T \Ac v + |v|^T \Ap |v|$.
Define a sign vector $s(v)\in\{-1,+1\}^n$ by $(s(v))_j = \operatorname{sign}(v_j)$---then $|v| = \diag(s(v)) v$ and
\begin{align*}
    &v^T Av \leq v^T \Ac v + v^T \diag(s(v)) \Ap \diag(s(v)) v  \\
    &= v^T (\Ac + \diag(s(v)) \Ap \diag(s(v))) v \\
    &\leq \max\{v^T (\Ac + \diag(s) \Ap \diag(s)) v : s\in\{-1,+1\}^n\}.
\end{align*}
Taking a $\max$ over $v$ with $|v|_2 = 1$ on either side and swapping the order of $\max$, 
$\mu_2(A) \leq \mu_2(\Ac + \diag(s) \Ap \diag(s))$ for some sign vector $s$. 
The containment $\Ac + \diag(s) \Ap \diag(s)\in[A]$ implies equality.
\end{proof}

\begin{remark}
In prior work \cite{davydov_verifying_2025}, a sufficient condition to check the negative (semi-)definiteness of an interval matrix $[A]$ constructs a matrix $B$ bounding the Metzler majorant $\lceil A\rceil_{\text{Mzr}}$\footnote{$\big(\lceil A\rceil_{\text{Mzr}}\big)_{ij} = |A_{ij}|$ if $i \neq j$ and $\big(\lceil A\rceil_{\text{Mzr}}\big)_{ii} = A_{ii}$.} of any matrix $A\in[A]$, and verifies that this matrix has all eigenvalues with negative real part. However, this sufficient condition is inherently conservative, even for single matrices. For example, $A = -tvv^T$ with $t > 0$ and $v = [1,\dots,1]^T \in \R^n$ is negative semidefinite yet its Metzler majorant has maximum eigenvalue equal to $t(n-2)$ which is positive for all $n > 2$ and grows unbounded as $t \to \infty$.
\end{remark}

Applying Lemma \ref{lem:corner_check} with an interval hull of the asymmetric contraction matrix over a region $X$, we can verify that $X$ is a contraction region through $2^n$ corner checks.

\begin{theorem} \label{thm:main_thm}
Let $[G]\in\IR^{n\times n}$ be an interval hull of the asymmetric contraction matrix $G$ from \eqref{eq:asym_contr_mtx} over $[0,\infty)\times X$, for fixed $c>0$ and region $X\subseteq\R^n$.
If 
\begin{align*}
    \max \left\{ \mu_2 (\Gc + \diag(s)\Gp \diag(s)) : s\in\{-1,+1\}^n \right\} \leq 0,
\end{align*}
then $X$ is a contraction region for the nonlinear system $\fpi$ from \eqref{eq:fpi} at rate $c$.
\end{theorem}
\begin{proof}
Lemma \ref{lem:corner_check} implies $\max_{x\in X} \mu_2(\change{G(t,x;a,c)}) \leq 0$, and Lemma \ref{lem:uniform_equivalence} completes the proof.
\end{proof}

\subsection{Loss for Certified Metric and Policy Training}

\begin{algorithm}[t]
\begin{algorithmic}[1]
\State \textbf{Input}: $a,b>0$ metric bounds, \change{$c>0$} contraction rate
\Function{loss}{$\Theta,\pi;X,a,b,c$} 
    \State $[G] \gets$ interval hull of $G$ from Definition \ref{def:asym_contr_mtx} over $X$ \label{alg:loss:Ghull} 
    \State $\lambda \gets \underset{s\in\{-1,+1\}^n}{\max}\mu_2(\Gc + \diag(s) \Gp \diag(s))$
    \State $[M] \gets$ interval hull of \change{$\Theta^T \Theta + aI$} over $X$ \label{alg:loss:Mhull} 
    \State $\hat{b} \gets \underset{s\in\{-1,+1\}^n}{\max}\mu_2(M^c + \diag(s) M_\Delta \diag(s))$
    \State \Return $\max\{\lambda,0\} + \max\{\hat{b} - b,0\}$
\EndFunction
\end{algorithmic}
\caption{Certified Neural Contraction Metric Loss}
\label{alg:loss}
\end{algorithm}

\begin{corollary} \label{cor:loss}
    If $\textsc{loss}(\Theta,\pi;X,a,b,c) \leq 0$ in Algorithm \ref{alg:loss}, for $a,b,c>0$, then
    \begin{enumerate}[label=(\roman*)]
        \item $M(x) = \Theta(x)^T\Theta(x)\change{+aI}$ satisfies $aI\preceq M(x) \preceq bI$;
        \item $X$ is a contraction region for the closed-loop system $\dot{x} = \fpi(x) = \fol(x,\pi(x))$ from \eqref{eq:fpi} at rate $c$.
    \end{enumerate}
\end{corollary}
\begin{proof}
    \change{$aI \preceq M(x)$ as $\Theta^T\Theta\succeq0$}.
    $\textsc{loss}(\Theta,\pi;X,a,b,c) \leq 0$ implies $\lambda \leq 0$ and $\hat{b} - b \leq 0$.
    By Lemma~\ref{lem:corner_check}, for every $M\in[M]$, 
    $\lambda_{\max}((M + M^T)/2) \leq \hat{b}$, implying part (i). 
    Part (ii) follows by application of Theorem \ref{thm:main_thm}.
\end{proof}

A crucial step in Algorithm \ref{alg:loss} is obtaining the interval hulls $[G]$ and $[M]$ containing the output of~\eqref{eq:asym_contr_mtx} over an input set $X$. 
In practice, we found that linear bound propagation in the style of CROWN \cite{zhang_efficient_2018,xu_automatic_2020} provides the right balance of accuracy, scalability, and trainability.
We use the implementation provided in our toolbox \immrax{} \cite{harapanahalli_immrax_2024} \footnote{\url{https://github.com/gtfactslab/immrax}} in JAX \cite{bradbury_jax_2018}. 
JAX allows us to automatically (i) construct the asymmetric contraction matrix \eqref{eq:asym_contr_mtx} using automatic differentiation, (ii) trace the computation graph applying linear bound propagation, and (iii) backpropagate the loss from Algorithm \ref{alg:loss} and dispatch the neural network training to the GPU.

\section{Explicit Tracking Control Design} \label{sec:tracking}

Consider the following control-affine nonlinear system,
\begin{align} \label{eq:control_affine}
    \dot{x} = \fol(x,u) = \fd(x) + g(x) u,
\end{align}
where $\fd:\R^n\to\R^n$ is a $C^1$ drift term and $g(x)\in\R^{n\times m}$ is a $C^1$ state dependent matrix.
With some restrictions on the choice of metric $M(x)$, we can apply the methodology from the previous section to design a simple feedback controller to exponentially track any dynamically feasible trajectory within the contraction region $X$.
\change{We assume this dynamically feasible trajectory is precomputed using a motion planning algorithm.}
Notably, we do not require any online geodesic computations like \cite{manchester_control_2017,singh_robust_2023}, or state augmentation with the trajectory to be tracked like previous learning methods \cite{sun_learning_2021}.

\begin{corollary} \label{cor:tracking}
Suppose $\textsc{loss}(\Theta,\pi;X,a,b,c) \leq 0$ for the nonlinear system $\dot{x} = \fd(x) + g(x)u$ with hyperparameters $a,b,c>0$, satisfying the following condition: for every $j=1,\dots,m$ and every $x\in X$,
\begin{align} \label{eq:killing_field}
    \change{\Theta(x)^T \partial_{g_j(x)} \Theta(x) + (\Theta(x)^T\Theta(x) + aI)\frac{\partial g_j}{\partial x} (x) = 0.}
\end{align}
Let $x'(t)$ be a trajectory associated with a nominal input curve $u'(t)$.
If there exists an $R>0$ such that the metric tube of radius $R$ around $x'(t)$ is fully contained inside $\intX$, \ie, for every $t\geq 0$, $\{x : d(x,x'(t)) \leq R\} \subseteq \intX$,
then any trajectory $x(t)$ under the feedback policy
\begin{align} \label{eq:tracking_control}
    u(t,x) = \pi(x) - \pi(x'(t)) + u'(t) 
\end{align}
with initial condition $d(x(0),x'(0)) \leq R$ satisfies, for $t\geq 0$,
\begin{align*}
    d(x(t),x'(t)) \leq e^{-ct} d(x(0),x'(0)).
\end{align*}
\end{corollary}

\begin{proof}
Letting $v(t) = u'(t) - \pi(x'(t))$, and $\fpi(t,x) = \fd(x) + g(x)\pi(x)$, the closed-loop dynamics are
\begin{align} \label{eq:fcl}
    \dot{x} = \fcl(t,x) 
    &= \fpi(t,x) + \textstyle\sum_{j=1}^m g_j(x) v_j(t).
\end{align}
Computing the asymmetric contraction matrix from Definition \ref{def:asym_contr_mtx}, by linearity,
for $t\in[0,\infty)$, $x\in X$,
\change{
\begin{align*}
    &\change{G_{\mathrm{cl}}(t,x;a,c)} = \change{G_{\pi}(t,x;a,c)} + \textstyle \sum_{j=1}^m \big[\Theta(x)^T \partial_{g_j(x)}\Theta(x) \\
    &\quad\quad\quad\quad \textstyle + \Theta(x)^T \Theta(x)\frac{\partial g_j}{\partial x}(x) 
    + a\frac{\partial g_j}{\partial x}(x) \big] v_j(t) 
\end{align*}
}
where \change{$G_{\pi}(t,x;a,c)$} is the asymmetric contraction matrix (Definition \ref{def:asym_contr_mtx}) of $\fpi$.
The condition $\textsc{loss}(\Theta,\pi;X,a,b,c) \leq 0$ implies that we have $\mu_2(G_\pi(t,x)) \leq 0$ by Corollary \ref{cor:loss}.
By assumption \eqref{eq:killing_field}, the bracketed term vanishes. 
Thus, subadditivity of $\mu_2(\cdot)$ implies $\mu_2(\change{G_{\mathrm{cl}}(t,x;a,c)}) \leq \mu_2(G_\pi(t,x)) + \mu_2(0) \leq 0$.
Therefore, Lemma \ref{lem:uniform_equivalence} implies $X$ is a contraction region for the closed-loop system $\fcl$ \eqref{eq:fcl} at rate $c$.

To conclude, since $u(t,x'(t)) = u'(t)$, the uniqueness of solutions implies $x'(t)$ is the trajectory of $\fcl$ from initial condition $x'(0)$.
Because $x'$ is a trajectory of a contracting system with a metric tube of radius $R$ around it contained in the open and connected set $\intX$, there is a minimizing geodesic between $x$ and $x'$ contained in $\intX$ which shrinks exponentially at rate $c$ \cite[Thm 1]{lohmiller_contraction_1998}.
\end{proof}

\begin{remark}[Geodesic metric tube]
The tube of radius $R$ in Corollary \ref{cor:tracking} is analytically important, ensuring geodesic paths connecting the two trajectories are always contained within the contraction region. 
However, computationally finding such sublevel sets of the geodesic metric is difficult due to the need to search for geodesics.
For simpler guarantees, one could apply the uniform bounds on the metric to obtain $\ell_2$-ball tubes sandwiching a geodesic metric tube:
if $\{x : \|x - x'(t)\|_2 \leq \olR\} \subseteq \intX$ for $\olR>0$, then the radius $R = \sqrt{a} \olR$ geodesic metric ball is also contained in $\intX$, and points satisfying $\|x - x'(0)\|_2 \leq \frac{\sqrt{a}}{\sqrt{b}} \olR$ live within the radius $R$ ball at time $0$.
In practice, we find that controllers trained to satisfy the contraction condition within region $X$ have good behavior when geodesics and even the trajectories themselves leave $X$ (see Section \ref{sec:example}).
\change{This may be due to the overestimation induced by bound propagation making the closed-loop system more robust than we can formally certify.}
\end{remark}

\begin{remark}[Killing field condition] \label{rem:Killing}
    The condition \eqref{eq:killing_field} is the asymmetric analogue to the strong CCM condition \cite[C2]{manchester_control_2017} which requires the actuation vector fields $g_j$ to be Killing fields of the metric $M(x) = \Theta(x)^T\Theta(x) \change{\, +aI}$,
\begin{align*}
    \textstyle \partial_{g_j(x)} M(x) + M(x) \frac{\partial g_j}{\partial x} (x) + \frac{\partial g_j}{\partial x}(x)^T M(x) = 0,
\end{align*}
which can be seen by a computation similar to Lemma \ref{lem:uniform_equivalence}.

For systems with a constant $g(x) = B$ for $B\in\R^{n\times m}$, we can augment the neural network with a linear layer projecting $x$ into the orthogonal subspace to $B$, using a matrix $B^\perp$ satisfying $B^\perp B = 0$.
For instance, when $B = [0, I]^T$, restricting the argument of $\Theta$ to the first $n-m$ components of $x$ structurally satisfies \eqref{eq:killing_field}.

In the general non-constant case, enforcing constraints such as this on the neural network $\Theta$ is generally challenging.
\change{With dynamic extension, any nonlinear system can be transformed into this form using an integral control law.}
\end{remark}

\begin{remark}[Control contraction metrics (CCM) and complete integrability]
Following the discussion in \cite[p. 3]{manchester_control_2017}, our learned metric $M(x) = \Theta(x)^T\Theta(x)$ and policy $\pi$ provide a \emph{completely integrable} differential control law $\delta u = \frac{\partial \pi}{\partial x}(x) \delta x$, which exponentially stabilizes the variational dynamics $\dot{\delta x} = (\frac{\partial \fd}{\partial x}(x) + \sum_{j=1}^m \frac{\partial g_j}{\partial x}(x) u'_j) \delta x + g(x)\delta u$, along trajectories $(x',u')$.
This can be seen by evaluating LMI~\eqref{eq:contraction_LMI} directly on the closed-loop dynamics $\fcl$ \eqref{eq:fcl}, and simplifying the resulting expression using the Killing field condition discussed in Remark \ref{rem:Killing} to obtain the following:
\begin{align*}
    M(x)& \textstyle \big(\frac{\partial \fd}{\partial x}(x) + g(x)\frac{\partial \pi}{\partial x}(x)\big) + (\star)^TM(x) \\
    &+ \partial_{\fd(x) + g(x)\pi(x)}M(x) + 2cM(x) \preceq 0,
\end{align*}
which verifies the quadratic (Finsler-)Lyapunov condition on $V(x,\delta x) = \delta x^T M(x) \delta x$.
In contrast, the convex criteria from CCM \cite{manchester_control_2017} instead require only a \emph{path integrable} differential controller $\delta u = K(t,x)\delta x$, where $K$ may not be the Jacobian of an explicit function.
This comes at the cost of an \emph{implicit} control policy, which requires an online search for a geodesic path between the current state and the reference.
Instead, our controller provides an \emph{explicit} representation, where two simple forward inferences of the policy $\pi$ provide the control input at minimal computational cost.
\end{remark}

\section{Example: Contracting Tracking Control for 10 State Quadrotor} \label{sec:example}

\begin{figure}
    \centering
    \begin{subfigure}{0.49\columnwidth}
        \includegraphics[width=\columnwidth, trim={0 2cm 0 3cm}, clip]{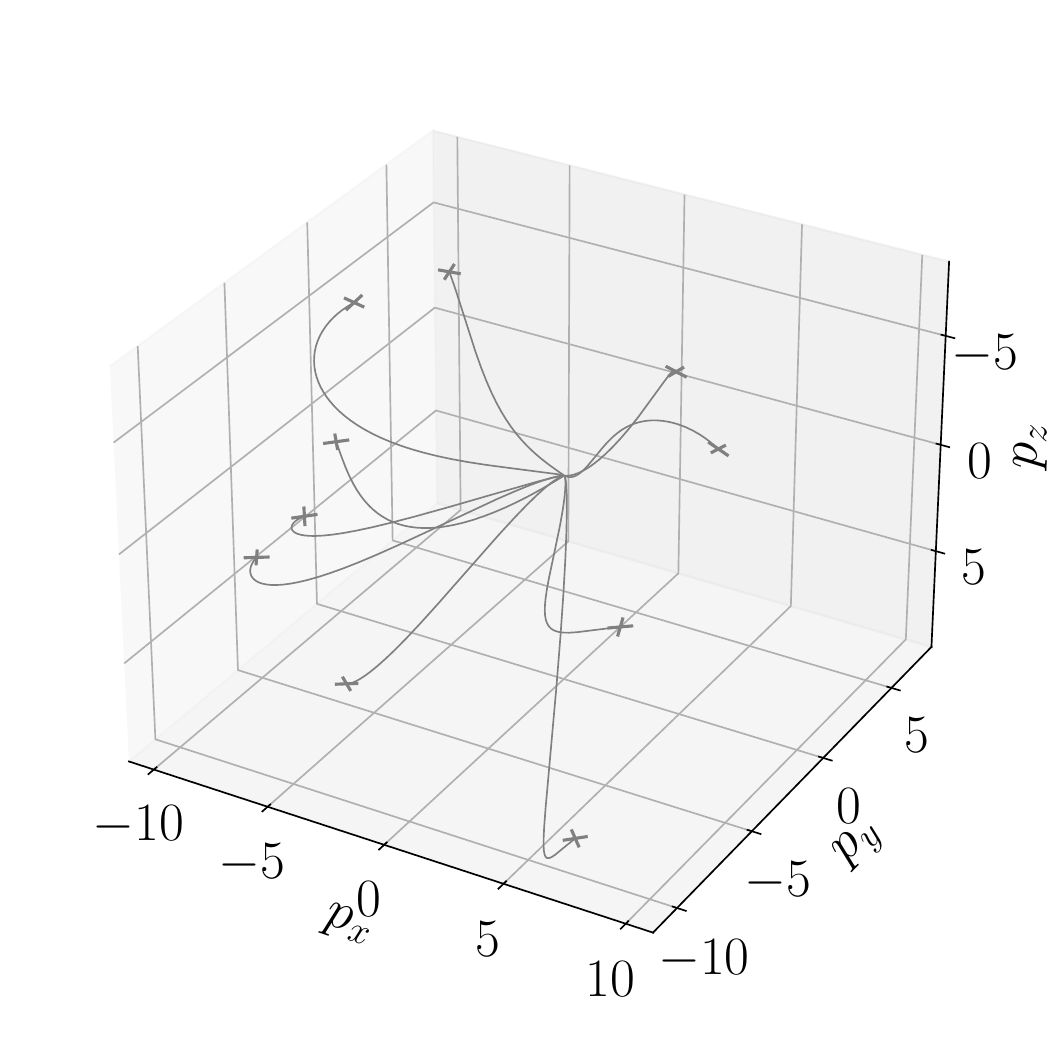}
        \caption{Hover}
        \label{subf:hover}
    \end{subfigure}
    \hfill
    \begin{subfigure}{0.49\columnwidth}
        \includegraphics[width=\columnwidth, trim={0 2cm 0 3cm}, clip]{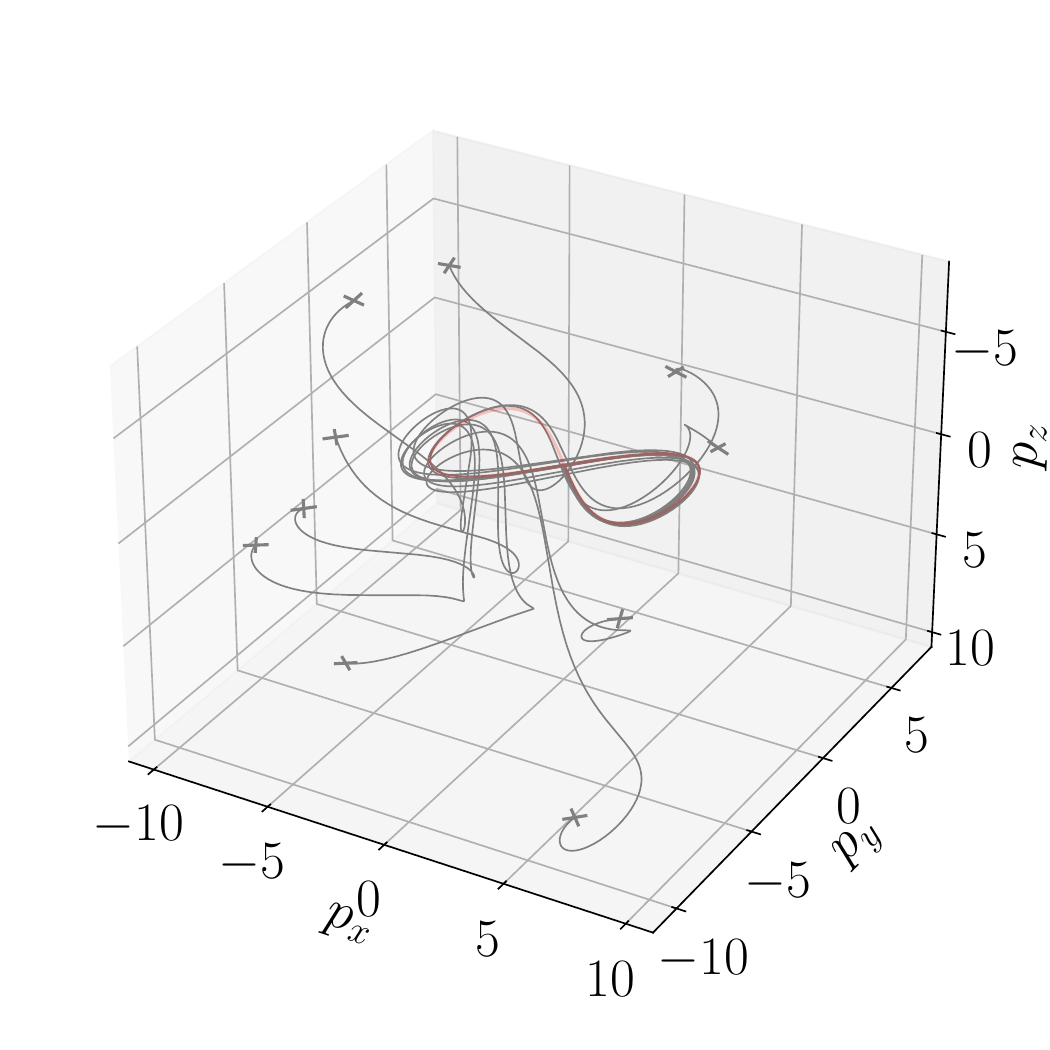}
        \caption{Figure Eight}
        \label{subf:eight}
    \end{subfigure}
    \begin{subfigure}{0.49\columnwidth}
        \includegraphics[width=\columnwidth, trim={0 2cm 0 3cm}, clip]{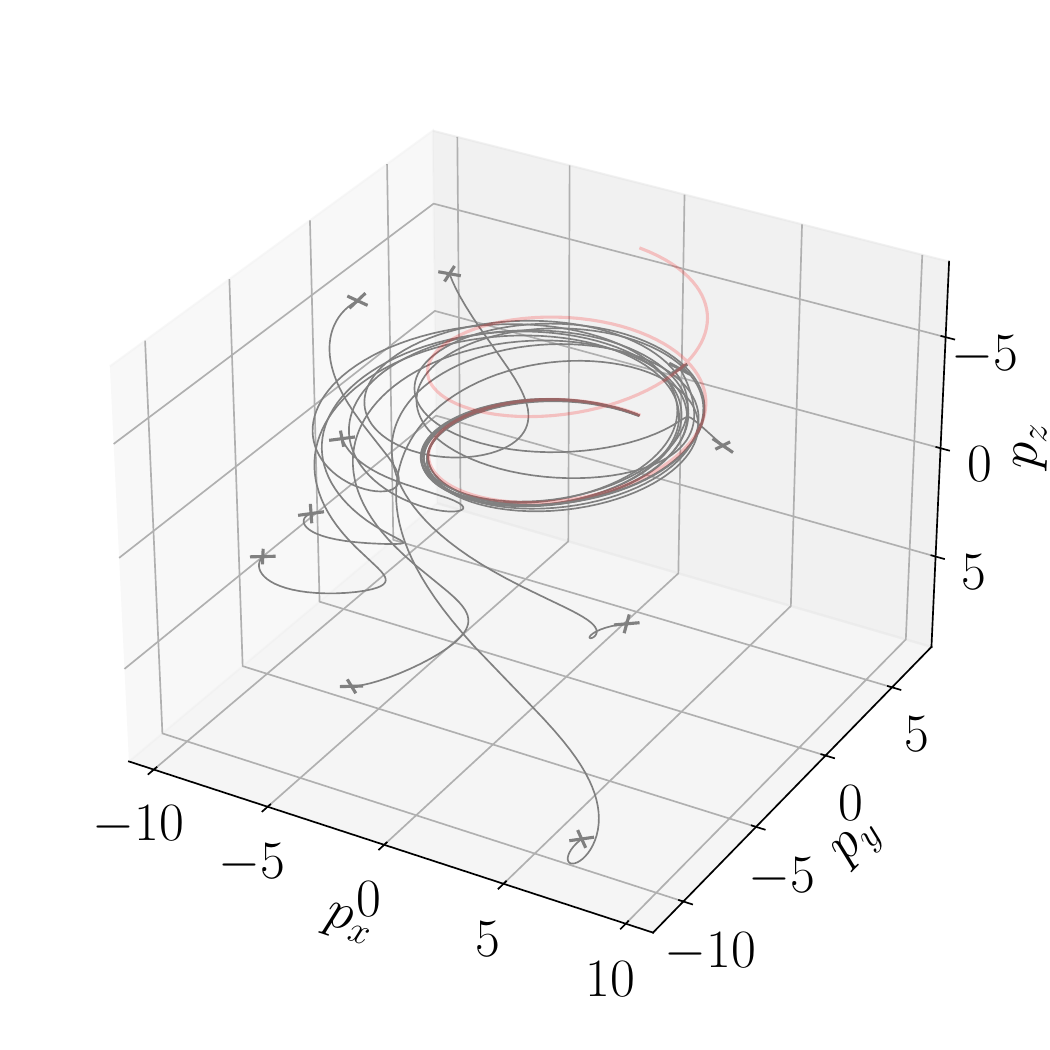}
        \caption{Helix}
        \label{subf:helix}
    \end{subfigure}
    \hfill
    \begin{subfigure}{0.49\columnwidth}
        \includegraphics[width=\columnwidth, trim={0 2cm 0 3cm}, clip]{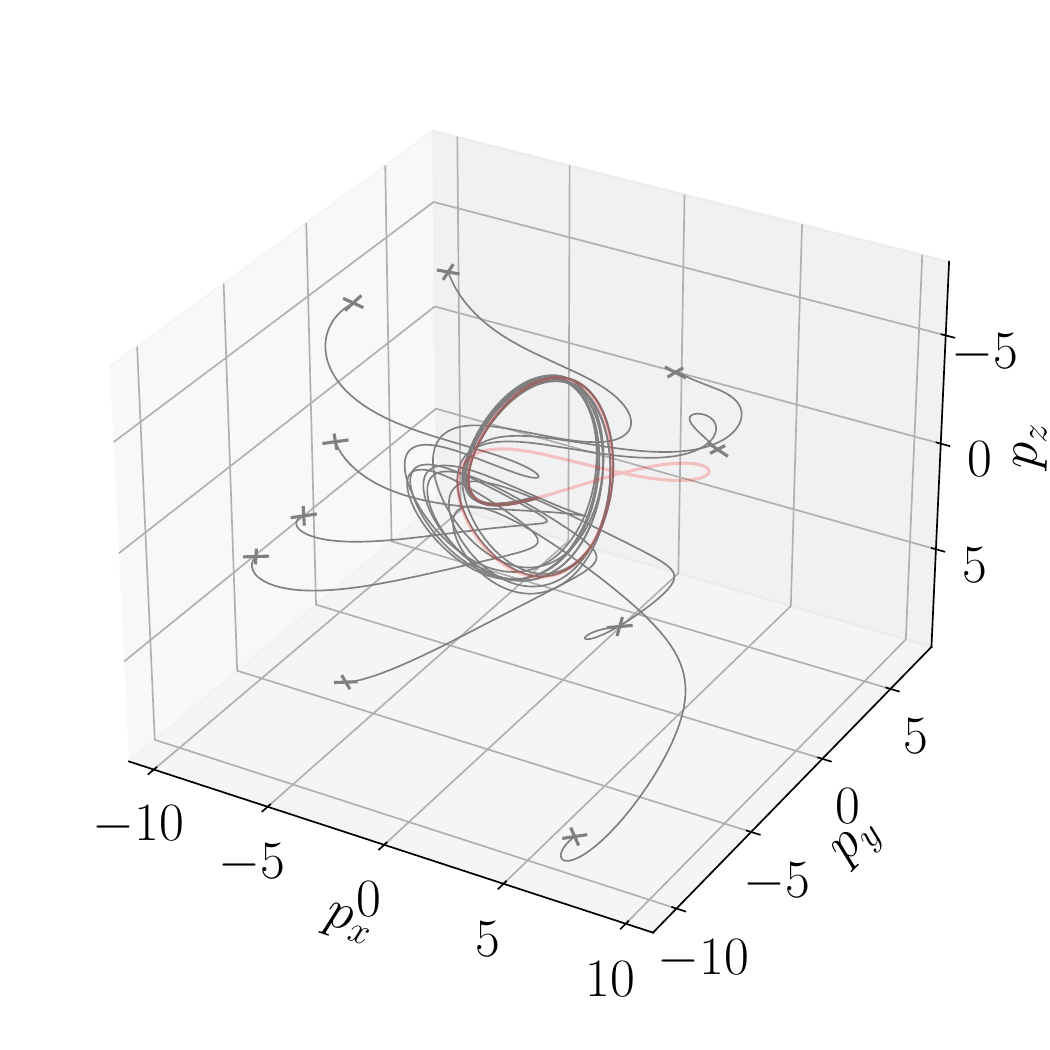}
        \caption{Trefoil Knot}
        \label{subf:trefoil}
    \end{subfigure}
    \begin{subfigure}{\columnwidth}
        \includegraphics[width=\columnwidth]{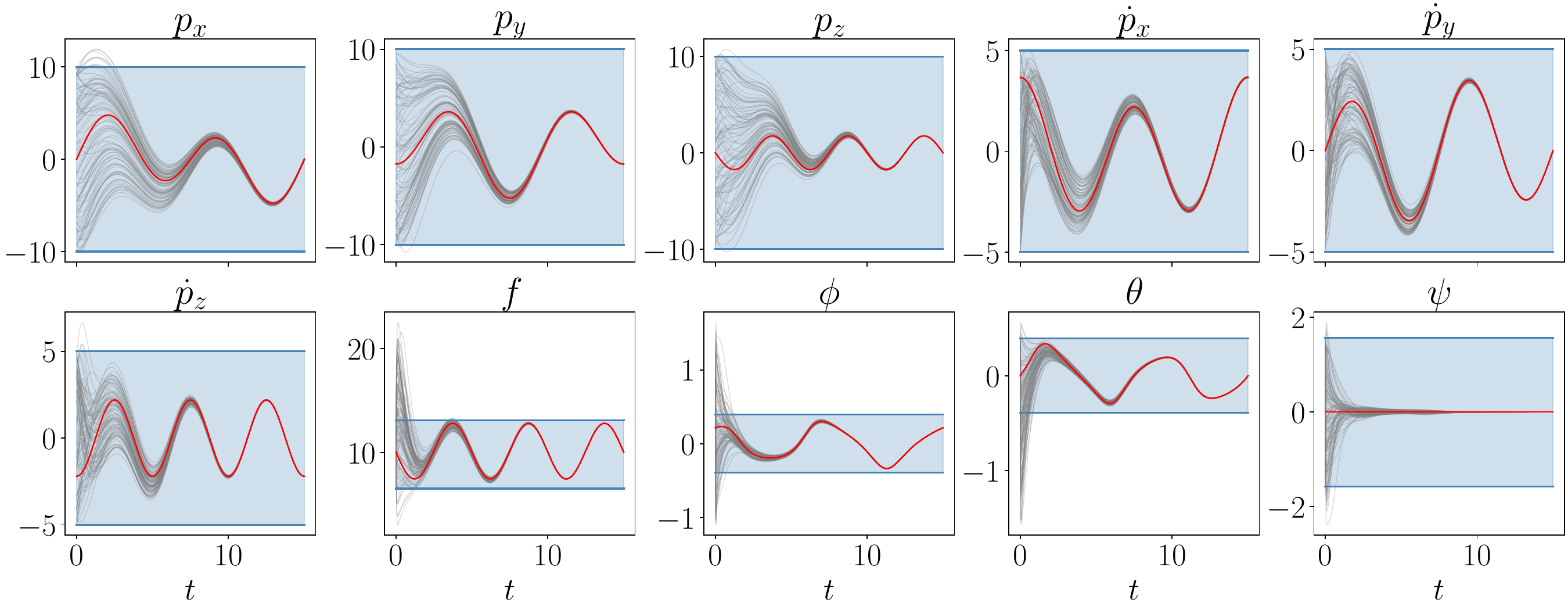}
        \caption{States vs Time for Trefoil Knot}
        \label{subf:states_time}
    \end{subfigure}
    \caption{
    Plots (a)-(d) \change{show the $(p_x,p_y,p_z)$ states for 10 randomly sampled initial conditions from the verified contraction region $X$ (pictured with the $\times$) track several nominal trajectories (pictured in red) under the feedback policy $u(t,x) = \pi(x) - \pi(x'(t)) + u'(t)$ from Corollary \ref{cor:tracking}.} 
    Plot (e) shows the state vs time for 100 randomly sampled initial conditions starting in $X$ for the trefoil knot (d). The nominal trajectory is in red, with the verified contraction region $X$ in blue. 
    While some trajectories leave the verified region, the controller empirically behaves well, bringing them back into the region and stabilizing to the commanded trajectory.
    }
    \vspace{-1em}
    \label{fig:quadrotor}
\end{figure}

In this section, we use Algorithm \ref{alg:loss} to design a contracting tracking controller for the 10-state quadrotor from \cite{singh_robust_2023}, \footnote{The code reproducing the results is available at \\ \url{https://github.com/gtfactslab/neural_contraction}}
\begin{align*}
    \ddot{p}_x = - \tau\sin\theta, \ 
    \ddot{p}_y = \tau \cos\theta\sin\phi, \
    \ddot{p}_z = g - \tau\cos\theta\cos\phi, 
\end{align*}
with $x = [p_x,p_y,p_z,\dot{p}_x,\dot{p}_y,\dot{p}_z,\tau,\phi,\theta,\psi]^T$, $[p_x,p_y,p_z]^T\in\R^3$ in North-East-Down convention, $[\phi,\theta,\psi]^T\in\R^3$ the XYZ Euler-angle representation, $\tau\in\R$ the net (mass-normalized) thrust, and $u = [\dot{\tau},\dot{\phi},\dot{\theta},\dot{\psi}]^T$ is the control input.

\change{For the controller, we learn a network with $2$ hidden layers of $32$ neurons returning a matrix $N(x) \in \R^{4\times11}$, and set $\pi(x) = N(x) \smallconc{x}{1}$ (affine parameterization).
For the metric, we learn a network with $2$ hidden layers of $32$ neurons returning the upper triangle of the matrix $\Theta(x) \in \R^{10\times 10}$.
}

We trained the contraction region $X = [-10,10]^3 \times [-5,5]^3 \times [\tfrac{2g}{3}, \tfrac{4g}{3}] \times [-\tfrac{\pi}{8},\tfrac{\pi}{8}]^2\times [-\tfrac{\pi}{2},\tfrac{\pi}{2}]$ by optimizing over growing regions $\frac{n}{100} X$, where we checkpoint and iterate $n$ when the loss hits $0$ until $n = 100$.
We partition each $n$-th region uniformly into \change{$10^2=100$ partitions over the $(\phi,\theta)$ states}.
We use the AdamW optimizer with the loss from Algorithm \ref{alg:loss}, with \change{$a=1$, $b=50$, $c=0.1$. 
In total, the training took a total of 10479 steps in 644 seconds, including 62 seconds for JIT compilation.}\footnote{Code was run on a desktop running Kubuntu 22.04, with AMD Ryzen 5 5600X, NVIDIA RTX 3070, and 32 GB of RAM.}
We use the differential flatness property of quadrotor dynamics to extract a dynamically feasible pair $(x'(t),u'(t))$ based on a desired $C^3$ curve $[p_x'(t),p_y'(t),p_z'(t),\psi'(t)]^T$. 
Then, we apply the explicit feedback policy derived in Corollary~\ref{cor:tracking} to track this trajectory.
Compared to \cite{sun_learning_2021}, we get a formal guarantee of contraction over the whole region $X$, which was \change{intractable} with sampling-based guarantees for $n=10$, and we completely avoid the need to search for minimizing geodesics like the approach from \cite{singh_robust_2023}.

\section{Conclusion and Future Work}

In this letter, we presented a framework for jointly learning a neural network controller and a neural contraction metric with formal guarantees of closed-loop contraction over a region $X$. 
By leveraging a novel asymmetric contraction condition, GPU-parallelized corner checks\change{, and end-to-end automatic differentiation}, our approach avoids \change{dependency management overconservativeness} and sample complexity issues associated to Lipschitz-based guarantees. 
For control-affine systems, we also developed an explicit tracking controller requiring only two forward inferences of the network.

Several directions remain for future work. Incorporating forward invariance guarantees for the certified region and enforcing bounded control inputs would strengthen the practical applicability of the framework. More broadly, extending the approach to learn Finsler-Lyapunov functions or to certify contraction with respect to non-Euclidean norms could further expand its scope to systems where Riemannian metrics are not the most natural choice.

\bibliographystyle{ieeetr}
\bibliography{references.bib}

\end{document}